\documentclass[11pt]{article}
    \author{Boris Aronov\thanks{Department of Computer Science and Engineering,
         Tandon School of Engineering, New York University, Brooklyn, NY~11201, USA.
         Email: \texttt{boris.aronov@nyu.edu}.
         BA has been supported by NSF grants CCF~15-40656 and CCF~20-08551, and by grant 2014/170 from the US-Israel Binational Science Foundation.}
        \and Abdul Basit\thanks{School of Mathematics,
                    Monash University, Australia. Email: \texttt{abdul.basit@monash.edu}.}
        \and Mark de Berg\thanks{Department of Computing Science, TU Eindhoven,
            P.O.~Box 513, 5600 MB Eindhoven, the Netherlands.
            Email: \texttt{mdberg@win.tue.nl}.
            MdB is supported by the Netherlands' Organisation for
            Scientific Research (NWO) under project no.~024.002.003.}
        \and Joachim Gudmundsson\thanks{School of Computer Science, University of Sydney.
              Email: \texttt{joachim.gudmundsson@sydney.edu.au}.
            JG is supported by the Australian Government through the Australian Research Council DP180102870.}
        }
\date{}
\title{Partitioning Axis-Parallel Lines in 3D} 

\usepackage{mathtools,amsthm,amssymb,xcolor,graphicx,tabularx,xspace,verbatim}
\usepackage{microtype}
\usepackage{hyperref}
\usepackage{cite}
\usepackage{enumerate}

    \usepackage[a4paper,includeheadfoot,margin=2.54cm]{geometry}
    \usepackage[small]{caption}

    \makeatletter
    \def\th@slanted{\sl }
    \makeatletter
    \theoremstyle{slanted}
    \newtheorem{theorem}{Theorem}
    \newtheorem{lemma}[theorem]{Lemma}
    
    \newtheorem{corollary}[theorem]{Corollary}
    \newtheorem{observation}[theorem]{Observation}
    
    \newtheorem{claim}[theorem]{Claim}
    \theoremstyle{definition}

\newenvironment{myquote}
  {\list{}{\leftmargin=5mm \rightmargin=5mm}\item[]}
  {\endlist}
  
\usepackage[shortlabels]{enumitem}

\renewcommand{\leq}{\leqslant}
\renewcommand{\geq}{\geqslant}
\renewcommand{\le}{\leqslant}
\renewcommand{\ge}{\geqslant}

\newenvironment{claiminproof}{\begin{myquote}\noindent\emph{Claim.}}{\end{myquote}}
\newenvironment{proofinproof}{\begin{myquote}\noindent\emph{Proof.}}{\hfill $\lhd$ \end{myquote}}

\newcommand{\Reals}{\mathbb{R}}
\newcommand{\Nats}{\mathbb{N}}

\newcommand{\A}{\ensuremath{\mathcal{A}}}

\newcommand{\NE}{\mbox{\sc ne}\xspace}
\newcommand{\NW}{\mbox{\sc nw}\xspace}
\newcommand{\SE}{\mbox{\sc se}\xspace}
\newcommand{\SW}{\mbox{\sc sw}\xspace}
\newcommand{\NM}{\mbox{\sc n-mid}\xspace}
\newcommand{\SM}{\mbox{\sc s-mid}\xspace}
\newcommand{\NEt}{\mbox{\sc ne-top}\xspace}
\newcommand{\NEb}{\mbox{\sc ne-bottom}\xspace}
\newcommand{\SEt}{\mbox{\sc se-top}\xspace}
\newcommand{\SEb}{\mbox{\sc se-bottom}\xspace}
\newcommand{\NWt}{\mbox{\sc nw-top}\xspace}

\newcommand{\SWt}{\mbox{\sc sw-top}\xspace}
\newcommand{\SWb}{\mbox{\sc sw-bottom}\xspace}

\newcommand{\total}{\mbox{\sc t}}
\newcommand{\north}{\mbox{\sc n}}
\newcommand{\east}{\mbox{\sc e}}
\newcommand{\south}{\mbox{\sc s}}
\newcommand{\west}{\mbox{\sc w}}

\newcommand{\Lfour}{L^4}

\DeclarePairedDelimiter\ceil{\lceil}{\rceil}
\DeclarePairedDelimiter\floor{\lfloor}{\rfloor}

\newcommand{\eps}{\varepsilon}
\newcommand{\mydef}{\coloneqq}

\newcommand{\gpar}{g_{\|}}
\newcommand{\gperp}{g_{\perp}}
\newcommand{\Ltb}{L^\mathrm{3}}
\newcommand{\Ltbx}{\Ltb_x}
\newcommand{\Ltby}{\Ltb_y}
\newcommand{\Ltbz}{\Ltb_z}

\begin{document}
\maketitle

\begin{abstract}
Let $L$ be a set of $n$ axis-parallel lines in $\Reals^3$. We are are interested in
partitions of $\Reals^3$ by a set $H$ of three planes such that each open cell in the
arrangement $\A(H)$ is intersected by as few lines from~$L$ as possible. We study
such partitions in three settings, depending on the type of splitting planes that
we allow. We obtain the following results.
\begin{itemize}
\item There are sets $L$ of $n$ axis-parallel lines such that, for any set $H$ of three
      splitting planes, there is an open cell in $\A(H)$ 
      that intersects at least~$\floor{n/3}-1 \approx  \frac{1}{3}n$ lines.
\item If we require the splitting planes to be axis-parallel, then there are 
      sets $L$ of $n$ axis-parallel lines such that, for any set $H$ of three
      splitting planes, there is an open cell in~$\A(H)$ that
      intersects at least~$\frac{3}{2}\floor{n/4}-1 \approx \left( \frac{1}{3}+\frac{1}{24}\right) n$ lines.
      
      Furthermore, for any set $L$ of $n$ axis-parallel lines, there exists a set $H$ of
      three axis-parallel splitting planes such that each open cell in $\A(H)$ 
      intersects at most $\frac{7}{18} n =  \left( \frac{1}{3}+\frac{1}{18}\right) n$ lines.
\item For any set $L$ of $n$ axis-parallel lines, there exists a set $H$ of
      three axis-parallel and mutually orthogonal splitting planes, such that each open cell in $\A(H)$ 
      intersects at most $\ceil{\frac{5}{12} n} \approx  \left( \frac{1}{3}+\frac{1}{12}\right) n$ lines.
\end{itemize}
\end{abstract}
\section{Introduction}
Partitioning problems of point sets in $\Reals^d$ have been studied extensively. For instance, the famous Ham-Sandwich Theorem states that,
given $d$ finite point sets in $\Reals^d$, there exists a hyperplane that bisects each of the sets, in the sense of having at most half of the set in each of its two open halfspaces.  Another well-known result is that for any set of $n$ points in the plane, there are two lines that partition the plane into four open cells that each contain at most $\floor{n/4}$ points (a \emph{4-partition}).  The latter result has several stronger forms, where one can specify the orientation of one of the lines, or that the two lines be orthogonal to each other~\cite{courant1942mathematics} (but not both).  The 4-partition question naturally generalizes to the problem of $2^d$-partitioning of a point set in~$\Reals^d$ by $d$~hyperplanes.  Such a triple of planes indeed always exists in $\Reals^3$~\cite{H66, Yao89}, in fact, with the orientation of one of the planes prespecified.  Alternatively, one of the partitioning planes can be required to be perpendicular to the other two~\cite{BK}.  It is known that $2^d$-partition does not always exist in $d>4$~\cite{avis1984}.  The case $d = 4$ remains stubbornly open. 
Results on other partitioning problems for finite point sets can be found in the surveys by Kaneko and Kano~\cite{kaneko2003} and by Kano and Urrutia~\cite{ku21}.

Similar theorems have been obtained for equipartitioning continuous measures. 
For example, the Ham-Sandwich theorem~\cite{roldan2021survey} is traditionally stated in continuous setting: Given $d$ finite absolutely continuous measures in $\Reals^d$, there exists a hyperplane that bisects each measure, in the sense of having half of the mass of each measure on each side; see \cite{BZ04} for some early history of the theorem. 
The $2^d$-partition problem mentioned above was originally asked by Gr{\"u}nbaum in 1960 \cite{grunbaum1960} for measures, as well; for a survey, see~\cite{blagojevic2018topology}.  For an overview of equipartitioning problems, see \cite{roldan2021survey,vzivaljevic2017}. 
\medskip

In this paper we are interested in the following question: given a set $L$ of $n$~lines,
partition the space into open cells such that each cell intersects only few lines. This problem has been studied extensively,
in the context of cuttings and polynomial partitions, but these works typically focus on 
asymptotic results that use a (possibly constant but) large number  of partitioning planes, or the zero set of a (constant but) large degree polynomial.  For example, a classical result on \emph{cuttings} in the plane~\cite{chazelle2018cuttings} is that, for any choice of parameter $r$, $1 \leq r \leq n$, there exists a tiling of the plane by $O(r^2)$ trapezoids so that each open trapezoid meets $n/r$ of the lines of~$L$.  In a similar spirit, a result of Guth~\cite{Guth} states that, for any degree $D>1$, there exists a non-zero bivariate polynomial $f$ of degree at most~$D$, such that the removal of its zero set $Z(f)$ from the plane produces $O(D^2)$ open connected sets, each meeting at most $n/D$ of the lines of $L$.  Analogous results are known for higher dimensions.

Another variant that has been studied is to partition $\Reals^3$ recursively using planes, until each cell meets $O(1)$ of the input objects. This results in a so-called binary space partition (BSP) of the objects. It has been shown that any set of lines (or disjoint triangles) admits a BSP of size $O(n^2)$~\cite{DBLP:journals/dcg/PatersonY90}, and any set of  axis-parallel lines (or disjoint axis-parallel rectangles) admits a BSP of size $O(n^{3/2})$~\cite{Paterson-Yao}. Both bounds are tight in the worst-case.

In contrast to the above settings, we are interested in what can be achieved with a very small number of planes, similar
to the results on Ham-Sandwich cuts and equipartitions. In particular, we are interested
in partitions for a set $L$ of $n$ lines in $\Reals^3$ that use only three planes.
More precisely, we want to partition $\Reals^3$ using a set $H$ of three planes such that 
each open cell in the arrangement $\A(H)$ meets only few of the lines from~$L$. Note that if the 
planes of $H$ are in general position, each such cell is an ``octant.''

We are not aware of any work directly addressing this question in three or higher dimensions, 
for unrestricted sets of lines.  The two-dimensional case was settled in \cite[Lemma 7]{MR2066580}: 
for any set of $n$ lines in the plane in general position, there exists a pair of lines with the 
property that each open quadrant formed by them is met by at most $\floor{n/4}$ of the lines.
In this initial study of the problem, we consider the case of axis-parallel lines in $\Reals^3$.

\paragraph{Problem statement and results.}
Let $L$ be a set of $n$ axis-parallel lines in $\Reals^3$.  
We would like to partition the space by three planes $h_1,h_2,h_3$, so that each of the open 
cells of $\A(H)$ meets as few lines of $L$ as possible, where $H\coloneqq \{h_1,h_2,h_2\}$.
We consider three variants, of increasing generality. The variants depend on the orientation of the
planes used in the partitioning. To this end we define a plane to be \emph{axis-parallel} if it
is parallel to two of the main coordinate axes, we define it to be \emph{semi-tilted} if 
it is parallel to exactly one of the main coordinate axis, and we define it to be
\emph{tilted} if not parallel to any coordinate axis. 
We will consider the following three types of partitionings.
\begin{itemize}[nosep]
\item The set $H$ must consist of three mutually orthogonal axis-parallel planes;  let $\gperp$ be the corresponding function.
\item The planes in $H$ must all be axis-parallel ($\gpar$), but need not be pairwise orthogonal.
\item The planes in $H$ can be chosen arbitrarily  ($g$), that is, $H$ can also use semi-tilted or tilted planes.
\end{itemize}
For a set $L$ of lines, we let $g(L)$ be the minimum integer $k$, such that there is a set~$H$ of three planes in $\Reals^3$ with every open cell of $\A(H)$ meeting at most $k$ lines of $L$.  Define $\gpar(L)$ and $\gperp(L)$ analogously.
Now define $g(n)\mydef\max_{|L|=n} g(L)$, with the maximum taken over all sets of axis-parallel lines of size $n$. 
Define $\gpar(n)$ and~$\gperp(n)$ similarly.
Clearly, $g(n)  \le \gpar (n) \le \gperp (n)$.  
Table~\ref{fig:results-table} summarizes our results.

\medskip
The rest of the paper is organized as follows.  In Section~\ref{sec:three-bundles} we introduce our simplest ``three bundle'' construction, which implies lower bounds on all three functions and illustrates some of the methods used in the rest of the paper.  We prove stronger lower bounds on $\gperp(n)$ and $\gpar(n)$ in Section~\ref{sec:lower-bounds}.
In Section~\ref{sec:upper-bounds} we present constructive upper bounds on $\gpar$ and $\gperp$, and consequently also on $g$.
\begin{table}
  \renewcommand{\arraystretch}{1.25} 
  \centering
  \begin{tabular}{l|ll|ll}
  Function & Lower bound & Reference & Upper bound & Reference \\
  \hline
  $\gperp(n)$ & 
  $\frac{3}{2}\left\lfloor \frac{n}{4} \right\rfloor \approx \left( \frac{1}{3} + \frac{1}{24} \right)n$
  & Theorem~\ref{th:unbalancedlb} & 
  $\left\lceil \frac{5}{12}n \right\rceil \approx \left( \frac{1}{3} + \frac{1}{12} \right)n$ & Theorem~\ref{thm:512upper-orig} \\[2mm]
  $\gpar(n)$ &   $\frac{3}{2}\left\lfloor \frac{n}{4} \right\rfloor-1 \approx \left( \frac{1}{3} + \frac{1}{24} \right)n$ & Theorem~\ref{th:unbalancedlb} & $\left\lfloor \frac{7}{18}n \right\rfloor \approx \left( \frac{1}{3} + \frac{1}{18} \right)n$ & Theorem~\ref{th:7nover18} \\[2mm]
  $g(n)$ & $\left\lfloor \frac{1}{3}n \right\rfloor -1  \approx \frac{1}{3}n$ & Corollary~\ref{cor:lb-general} &
  $\left\lfloor \frac{7}{18}n \right\rfloor \approx \left( \frac{1}{3} + \frac{1}{18} \right)n$ & Theorem~\ref{th:7nover18} \\
  \end{tabular}
  \caption{Summary of general results}
  \label{fig:results-table}
\end{table}

\medskip

We will use the following convention throughout the paper: we say that a line $\ell$~\emph{crosses} a plane~$h$ in~$\Reals^3$, if $\ell$~meets~$h$, but is not contained in it.

\section{Three bundles}
\label{sec:three-bundles}
As a warm-up exercise we analyze a simple and natural configuration of lines.
The analysis will give a first lower bound on the functions $\gperp(n)$, $\gpar(n)$, and $g(n)$.
In fact, the bound we obtain for $g(n)$ is the best bound we have for this case,
where the splitting planes can be completely arbitrary. 
The configuration consists of three bundles of $n/3$ lines, where $n$ is a multiple of~3;
see Fig.~\ref{fig:bundles-construction} for an illustration.
\begin{figure}
\begin{center}
\includegraphics{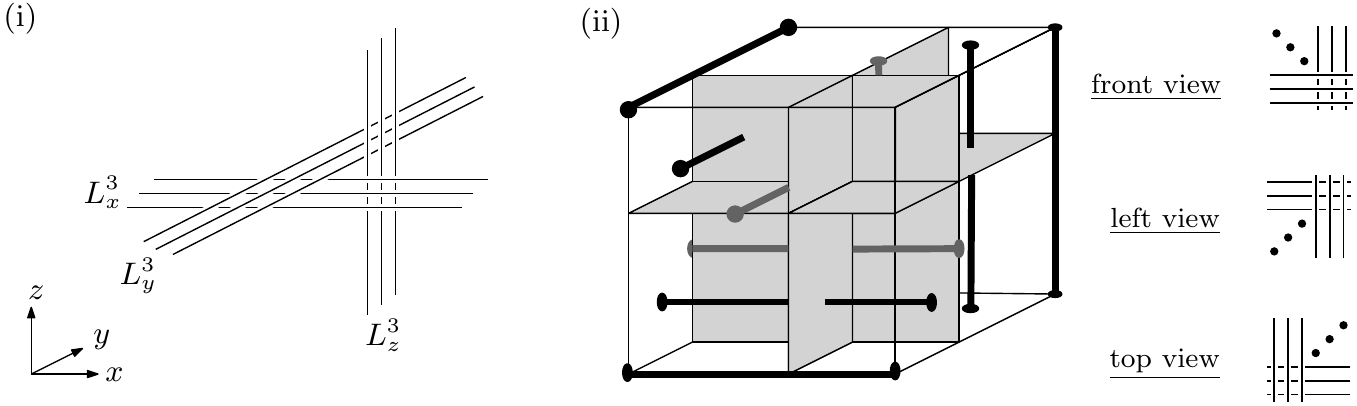}
\end{center}
\caption{(i)~Rough illustration of the set $\Ltb$ for the case $n=9$. 
         (ii)~A more accurate illustration of~$\Ltb$, which also shows three splitting
         planes that partition space into eight cells, each containing at most two lines.
         The dark grey lines are contained in a splitting plane.
         The front, left, and top views show that the construction is symmetric
         with respect to the axes, up to reversing the directions of some axes.}
\label{fig:bundles-construction}
\end{figure}
We call this configuration the \emph{three-bundle configuration}, and denote it by~$\Ltb(n)$, or just $\Ltb$ for short. 
It is defined as $\Ltb \mydef \Ltbx \cup \Ltby \cup \Ltbz$, where 
\begin{align*}
\Ltbx & \mydef \{ (0,i,i) + \lambda (1,0,0): 1\leq i \leq n/3 \},\\
\Ltby & \mydef \{ (i,0,2n/3+1-i) + \lambda (0,1,0):  1\leq i \leq n/3 \}, \text{ and} \\
\Ltbz & \mydef \{ (i,i,0) + \lambda (0,0,1): n/3 < i \leq 2n/3 \}.
\end{align*}
To be able to handle the case of tilted splitting planes, we need to slightly shift the lines
in~$\Ltb$, so that no plane (titled or otherwise) contains three or more lines.
For example, $\Ltbx$ is actually defined as  
$\Ltbx \mydef \{ (0,i+\eps_i,i) + \lambda (1,0,0): 1\leq i \leq n/3 \}$,
where the $\eps_i$'s are small, distinct real numbers that guarantee that
no plane contains more than two of the lines from~$\Ltbx$. 
The sets $\Ltby$ and $\Ltbz$ are shifted similarly. Thus no plane
can contain more than two lines from the same bundle. 
Note that a plane cannot contain two lines from different bundles either.
With a slight abuse of notation we will still
denote the shifted sets by $\Ltbx$, $\Ltby$, and $\Ltbz$. For simplicity,
we will refrain from showing the perturbations in our figures.
\medskip

It is straightforward to verify that $\gperp(\Ltb) \leq n/3 -1$. Indeed, if we take
$h_1 \mydef \{x=n/3+1\}$, and $h_2 \mydef \{y=n/3\}$, and $h_3 \mydef \{z=n/3+1\}$,
then six of the eight cells in the resulting arrangement $\A(H)$ are
intersected\footnote{Recall that we consider the cells to be open,
so lines fully contained in a splitting plane do not contribute to the counts.}
by $n/3-1$ lines (all coming from a single bundle)  and the remaining
two cells are not intersected at all; see Fig.~\ref{fig:bundles-construction}(ii).
The main result of this section is that this is tight: even if one
is allowed to use tilted splitting planes, it is not possible to ensure
that all cells in the partitioning are intersected by strictly fewer than~$n/3-1$ lines.
This gives the following theorem.
\begin{theorem}\label{thm:3-bundles}
Let $\Ltb(n)$ be the three-bundle configuration on $n$ lines, where $n$ is divisible by $3$, 
as defined above. Then 
\begin{itemize}[nosep]
\item $\gperp(\Ltb(n)) = \gpar(\Ltb(n)) = g(\Ltb(n)) = n/3-1$ for $n\geq 9$,
\item $\gperp(\Ltb(n)) = \gpar(\Ltb(n)) =1 $ and $g(\Ltb(n)) = 0$ for $n=6$,
\item $\gperp(\Ltb(n)) = \gpar(\Ltb(n)) = g(\Ltb(n)) = 0$ for $n=3$.
\end{itemize}
\end{theorem}
Theorem~\ref{thm:3-bundles} immediately gives the following corollary.
The cases where $n=3$ and $n=6$ are easy to verify.  In the remainder of this section we assume $n\geq 9$. 
\begin{corollary}
  \label{cor:lb-general}
  For $n\geq 9$, $g(n) \geq \floor{n/3} -1$.
\end{corollary}
Theorem~\ref{thm:3-bundles} also implies that $\gperp(n) \geq  \floor{n/3} -1$
and $\gpar(n) \geq \floor{n/3} -1$, of course, but in later sections we will
provide stronger lower bounds on $\gperp(n)$ and $\gpar(n)$; see Theorem~\ref{th:unbalancedlb}.
\medskip

Since we already argued that $\gperp(\Ltb) \leq n/3-1$, and we have $\gperp(\Ltb) \geq \gpar(\Ltb) \geq g(\Ltb)$,
it suffices to prove that $g(\Ltb)\geq n/3-1$. It will be convenient, however, to 
first prove this for $\gperp(\Ltb)$, then extend the proof to $\gpar(\Ltb)$,
and finally extend it to $g(\Ltb)$.
\begin{lemma}\label{lem:3b-lb-perp}
$\gperp(\Ltb(n)) \geq n/3-1$. 
\end{lemma}
\begin{proof}
Suppose for a contradiction that there is a set $H = \{h_1,h_2,h_3\}$ of mutually 
orthogonal axis-parallel planes such that each cell in the arrangement $\A(H)$ meets 
fewer than $n/3-1$~lines from~$\Ltb$. Let $h_1, h_2$, and $h_3$ be the planes orthogonal 
to the $x$-, $y$-, and $z$-axis, respectively.
\begin{figure}
\centering
\includegraphics[scale=0.9]{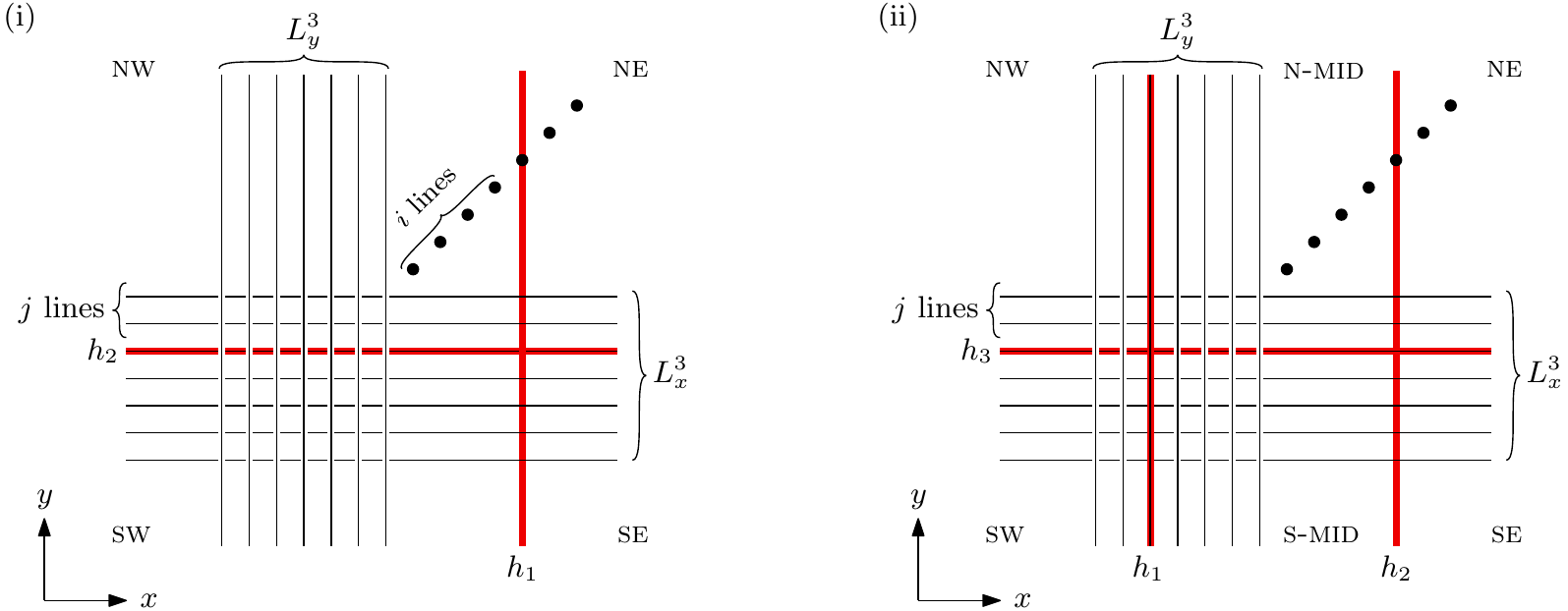}
\caption{Illustration for the proof of Lemmas~\ref{lem:3b-lb-perp} and~\ref{lem:3b-lb-par},
         showing the projection onto the $xy$-plane, where above/below crossings are drawn as seen from $z=+\infty$.}
\label{fig:lemma3and4}
\end{figure}

Consider the projection of $\Ltb$ and the planes $h_1$ and $h_2$ onto the $xy$-plane; refer to Fig.~\ref{fig:lemma3and4}(i).
The lines from~$\Ltbx$ and~$\Ltby$ appear as lines
in the projection. The planes $h_1$ and $h_2$ appear as lines as well, and they partition 
the $xy$-plane into four quadrants, which we label as \NE, \SE, \SW, and \NW in the natural way.
The lines from $\Ltbz$ appear as points in the projection. Below we will, with a slight abuse 
of terminology, sometimes speak of ``points from $\Ltbz$.''

Note that the \NE-quadrant contains at most $n/3-2$ points from $\Ltbz$; 
otherwise there would be two octants in $\A(H)$ meeting at least $n/3-1$ lines. 
Hence, at least one of the following two cases occurs
\begin{enumerate}[(a),nosep]
\item at least one point from $\Ltbz$ lies to the left of $h_1$, or
\item at least one point from $\Ltbz$ lies below $h_2$. 
\end{enumerate}
Due to symmetry we can assume without loss of generality that case~(a) holds,
which implies that all lines from $\Ltby$ lie to the left of $h_1$.
Now observe that the \SW-quadrant meets at most $2n/3-3$ lines from $\Ltbx \cup \Ltby$; 
otherwise one of the two octants in $\A(H)$ that correspond to the \SW-quadrant---these octants 
are generated when the $z$-vertical column corresponding to 
the \SW-quadrant is split by the splitting plane~$h_3$---will
meet at least $n/3-1$ lines. (Note that the plane $h_3$ can contain at most one line from $\Ltbx \cup \Ltby$.)
Since we already observed that all lines from $\Ltby$ lie to the left of $h_1$, 
this implies that at least two lines from $\Ltbx$ lie above~$h_2$. 
Hence, we can assume that the projection looks like the one depicted in Fig.~\ref{fig:lemma3and4}(i).
    
Let $i$ be the number of lines from $\Ltbz$ lying (strictly) to the left of $h_1$,
and let $j$ be the number of lines from $\Ltbx$ lying (strictly) above $h_2$;
by the above arguments we have $i\geq 1$ and $j\geq 2$.
Consider the remaining cutting plane~$h_3$, which is parallel to the $xy$-plane, and let $k$ be the number of lines from $\Ltby$ lying (strictly) below~$h_3$.
We have $k\geq 1$ because otherwise the
\NWt octant of $\A(H)$---this is the cell that is above $h_3$ and whose
projection onto the $xy$-plane is the \NW~quadrant---will meet $i$ lines of $\Ltbz$ and at least
$n/3-1$ lines of $\Ltbx$, which is a contradiction. 
Since $k\geq 1$ we also know that all $j$~lines from $\Ltbx$ must lie below $h_3$,
since these lines lie below the lines from $\Ltby$.    
We can thus derive the following. 
\begin{enumerate}[(i),nosep]
\item The \NEb octant meets $n/3-i-1$ lines of $\Ltbz$ and $j$ lines
    of $\Ltbx$. Since we assumed all cells intersect fewer than $n/3-1$ lines,
    we thus have $n/3-i-1+j<n/3-1$, which implies $i>j$. 
\item The \NWt octant meets $n/3-k-1$ lines of $\Ltby$ and $i$ lines of $\Ltbz$. 
   For this octant to intersect fewer than $n/3-1$ lines, we must thus have  $k>i$. 
\item The \SWb octant meets $k$ lines of $\Ltby$ and $n/3-j-1$ lines of $\Ltbx$. 
   For this octant to intersect fewer than $n/3-1$ lines, we must thus have $k<j$.
\end{enumerate}
But (i) and (ii) together imply $k>j$, which contradicts~(iii), thus finishing the proof of the lemma.
\end{proof}
We now consider the case where the planes, though still axis-parallel, need not be
mutually orthogonal.
\begin{lemma}\label{lem:3b-lb-par}
$\gpar(\Ltb(n)) \geq n/3-1$.
\end{lemma}
\begin{proof}
Suppose for a contradiction that there is a set $H = \{h_1,h_2,h_3\}$ of three axis-parallel 
planes such that each cell of $\A(H)$ meets fewer than $n/3-1$~lines from~$\Ltb$. 

The case of three mutually orthogonal planes has been handled in Lemma~\ref{lem:3b-lb-perp}.
If $h_1$, $h_2$, and~$h_3$ are all parallel to each other, then clearly there is a cell that 
is intersected by at least $n/3$ lines from $\Ltb$, 
since any line from the bundle orthogonal to the planes intersects all four cells.
Thus the remaining case is when exactly two of the planes are parallel. 

Assume without loss of generality that $h_1$ and $h_2$ are parallel to the $yz$-plane, 
with $h_2$ being to the right of $h_1$, and that $h_3$ is parallel to the $xz$-plane.
Now consider the situation in the projection onto the $xy$-plane, as shown in Fig.~\ref{fig:lemma3and4}(ii). 
There are six cells in the projection, which we label as
\NE, \NM, \NW, \SE, \SM, and \SW. 
If each cell in the projection meets fewer than $n/3-1$ lines, then at least one line of 
$\Ltby$ (in the projection) should be above $h_3$, 
otherwise the cells below~$h_3$ intersect more than $n/3-1$ lines.

Let $j\geq 1$ be the number of lines from $\Ltbx$ above~$h_3$.
Then the number of lines from $\Ltbz$---these are points in the figure---that are to the right of $h_2$ is 
at most~$n/3-j-2$, otherwise the \NE cell meets at least $n/3-1$ lines.
Similarly, the number of lines from $\Ltby$ to the left of $h_1$ must be 
at most~$j-1$, otherwise the \SW cell meets at least 
$(n/3-j-1) +  j = n/3-1$ lines. 
But then we have
\begin{itemize}[nosep]
\item $\mbox{(\# lines from $\Ltbx$ in \NM)} = j$,
\item $\mbox{(\# lines from $\Ltby$ in \NM)} \geq n/3 - (j-1)-1 = n/3 - j$,
\item $\mbox{(\# lines from $\Ltbz$ in \NM)} \geq n/3 - (n/3-j-2) -1 =  j+1$.
\end{itemize}
This brings the total number of lines in the \NM cell to at least $n/3 + j + 1$,
thus contradicting our assumption. \qedhere
\end{proof}
We now turn our attention to the general case, where we are also allowed to use
splitting planes that are tilted or semi-tilted.  (Recall that a semi-tilted plane is parallel to 
exactly one coordinate axis, and a tilted plane is
not parallel to any coordinate axis.)  
\begin{lemma}\label{lem:3b-lb-general}
  For $n\geq 9$, we have $g(\Ltb(n)) \geq n/3-1$.
\end{lemma}
\begin{proof}
Consider three splitting planes $h_1$, $h_2$ and $h_3$, and suppose for a contradiction that
each cell in $\A(H)$ intersects fewer than~$n/3-1$ lines. Define
\begin{itemize}[nosep]
\item $a \mydef \text{number of axis-parallel planes in $H$}$,
\item $s \mydef \text{number of semi-tilted planes in $H$}$, and
\item $t \mydef \text{number of tilted planes in $H$}$.
\end{itemize}
Note that $a+s+t=3$. The case where all three planes are axis-parallel has 
already been handled in the previous two lemmas, so we can assume that $a<3$. 
Any axis-parallel plane is crossed by $n/3$ lines from $\Ltb$, any semi-tilted
plane is crossed by $2n/3$ lines from $\Ltb$, and any tilted plane is crossed by $n$ lines from $\Ltb$.
(Recall that, when we say that a line \emph{crosses} 
a plane, or a plane \emph{crosses} a line, we mean that there is a proper intersection; 
that is, the line is not contained in the plane.) 
Thus 
\[
\text{(\# fragments generated by the planes in $H$)} 
= \frac{n}{3} \cdot (3 +a+2s+3t).
\]
Some fragments may be contained in one of the splitting planes and, hence, not
appear inside any cell. Note that an axis-parallel splitting plane may contain 
at most one line from $\Ltb$, a semi-tilted splitting plane may contain at most two lines
from $\Ltb$, and a tilted splitting plane cannot contain any line from $\Ltb$.
Furthermore, a line contained in splitting plane consists of at most three fragments,
arising due to the line crossing the other two splitting planes. Thus we have
at least $(n/3) \cdot (3+a+2s+3t) - 3(a+2s)$ fragments
appearing inside the cells in~$\A(H)$. Rewriting this, we conclude that
\begin{equation}
\text{(\# cell-line intersections)} \ge  (n/3) \cdot (6+s+2t) - 3(a+2s).  \label{eq:intersections}
\end{equation}
Since $\A(H)$ has at most eight cells, there must be a cell with 
$
\ceil*{\frac{(n/3) \cdot (6+s+2t) - 3(a+2s)}{8}}
$
intersections.
We now make a case distinction, depending on the values of $a$, $s$ and $t$.
Recall that $n$ is a multiple of~3 and that we are now dealing with the case~$n\geq 9$.
The first five cases we consider can be handled easily using Eq.~\eqref{eq:intersections},
as follows. 
\begin{itemize}
    \item If $t\geq 2$, or $t=1$ and $s=2$, there is a cell with at least 
      $\ceil*{\frac{(n/3)\cdot 10 - 6}{8}} \geq n/3-1$ fragments.
\item If $t=s=1$  (hence, $a=1$) there is a cell with at least 
      $\ceil*{\frac{(n/3)\cdot 9 - 9}{8}}
       = \ceil*{\frac{9}{8}(n/3-1)}
       \geq n/3-1$ fragments.
\item If $t=1$ and $s=0$ (hence, $a=2$) there is a cell with at least 
      $\ceil*{ \frac{(n/3)\cdot 8 - 6}{8}}
      = \frac{n}{3}-1$ fragments.
\item If $s=3$ (hence, $t=a=0$) there is a cell with at least 
      $\ceil*{\frac{(n/3)\cdot 9 - 18}{8}}
      = \ceil*{\frac{n}{3}-1 + \frac{n-30}{24}}$ fragments.  
      For $n\geq 9$ this is at least $n/3-1$.
\item If $t=0$ and $s=2$ (hence, $a=1$) there is a cell with at least 
      $\ceil*{\frac{(n/3)\cdot 8 - 15}{8}}
       = \ceil*{n/3- 15/8}
       = n/3-1$ fragments.
\end{itemize}
Each of the cases above gives us the desired contradiction.
The only remaining, and most difficult, case is when $t=0$, $s=1$, and $a=2$.
In this case we need a more refined analysis, given next.
\medskip

In the rest of the proof we assume there are two axis-parallel splitting planes, 
say $h_1$ and $h_2$, and one semi-tilted splitting plane $h_3$.
If $h_1$ and $h_2$ are parallel to each other, the number of cells in $\A(H)$ is six.
If all three planes are parallel to the same axis, then $\A(H)$ has seven cells.
In either case, from Equation~\ref{eq:intersections} it follows that there is a cell with
\[
\ceil*{\frac{(n/3) \cdot (6+s+2t) - 3(a+2s)}{7}}
= 
\ceil*{\frac{(n/3) \cdot 7- 12}{7}}
= \ceil{n/3 - 12/7} \geq 
n/3 -1.
\]
intersections giving the desired contradiction.  
Recall that the construction is symmetric with respect to the 
axes---see Fig.~\ref{fig:bundles-construction}(ii)---and so without loss of generality 
we may now assume that $h_1$ is perpendicular to the $x$-axis, 
$h_2$ is perpendicular to the $y$-axis, and $h_3$ is parallel to the $y$-axis.
We now proceed to derive a contradiction with our assumption that each cell
in $\A(H)$ meets less than $n/3-1$ lines. 
To this end we first prove several properties that the splitting planes~$h_1,h_2,h_3$
must satisfy, under our assumption. In the following, we consider the projection onto
the $xy$-plane, and statements like ``to the left of'' or ``above'' refer to the
situation in this projection. See Figure~\ref{fig:semi-tilted-1}(i), and note that the lines
of $\Ltbz$ show up as points in the figure.
\begin{figure}
\begin{center}
\includegraphics{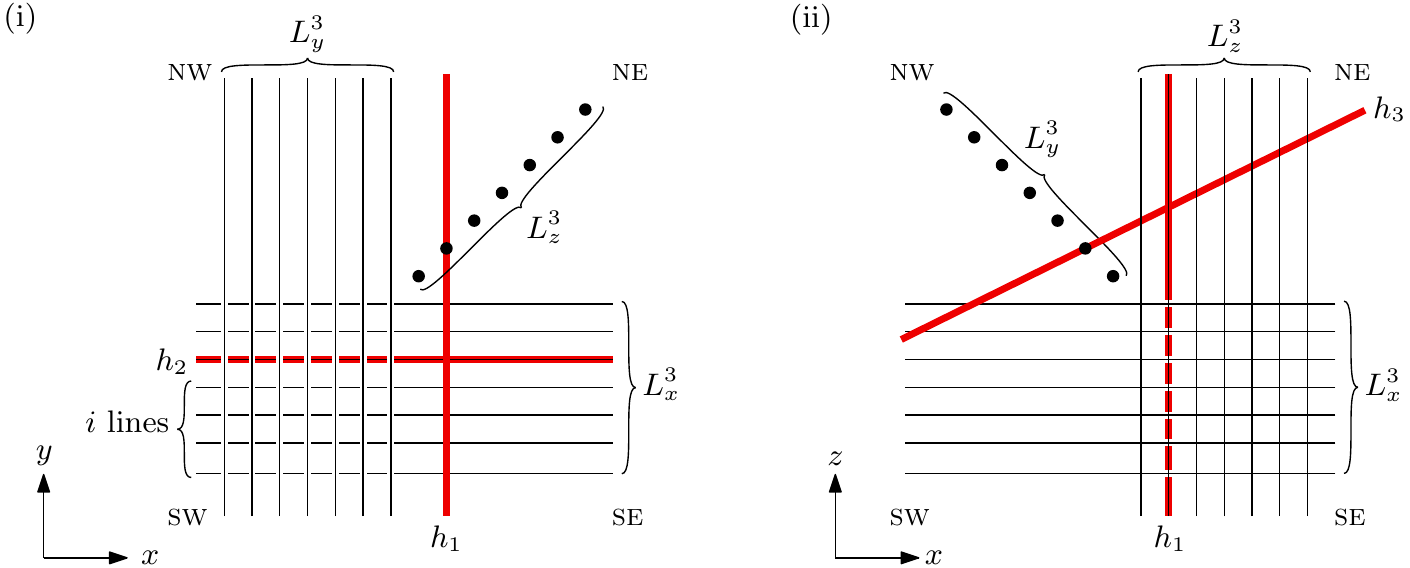}
\end{center}
\caption{Illustration for the proof of Lemma~\ref{lem:3b-lb-general}. Note that in
part~(ii) of the figure, the semi-tilted plane~$h_3$ can also show up as a line with negative 
slope. However, due to the small perturbation of the lines in $\Ltby$, $h_3$ can contain 
at most two lines from $\Ltby$ (here showing up as points).
}
\label{fig:semi-tilted-1}
\end{figure}
\begin{myquote}   
\vspace*{-5mm}
\begin{claim} \label{claim:1}
At least one line of $\Ltbx$ must lie above $h_2$.
Hence, all lines of $\Ltbz$ lie above~$h_2$.
\end{claim}
\begin{proof}
  For a contradiction, assume that no line of $\Ltbx$ lies strictly above $h_2$. 
  There is at most one line contained in $h_2$, so at least $n/3-1$ lines of $\Ltbx$
  lie below~$h_2$. Each such line crosses $h_1$ and $h_3$---recall that the latter splitting
  plane is semi-tilted and parallel to the $y$-axis---and therefore crosses three of the
  four octants below $h_2$. Moreover, each line of $\Ltby$, except at most
  one line contained in~$h_1$ and at most two lines contained in~$h_3$, intersects one of these octants. 
  Hence, the total count for the four octants below $h_2$ is at least 
  $3\cdot (n/3-1) + (n/3-3)$, and consequently, at least one of the four octants must intersect at least
  \[
  \ceil*{ \frac{3\cdot (n/3-1) + (n/3-3)}{4} }
  = 
  \ceil*{ \frac{4n/3-6}{4} }
  = 
  n/3-1
  \]
  lines, contradicting our assumption. (Here we use our assumption 
  that $n$ is a multiple of~3 and that $n\geq 9$.)
\end{proof}
\end{myquote}
Claim~\ref{claim:1} can be used to restrict the possible positions of~$h_1$.
\begin{myquote}   
\vspace*{-5mm}
\begin{claim} \label{claim:2}
$h_1$ must have at least one line of $\Ltbz$ on either side of it.
Hence, all lines of $\Ltby$ lie to the left of~$h_1$.
\end{claim}
\begin{proof}
By Claim~\ref{claim:1}, all lines from $\Ltbz$ lie above~$h_2$. If $h_1$ does not
have lines from $\Ltbz$ on either side of it, then one side has at least $n/3-1$ of these lines.
Hence, one of the four quadrants defined by $h_1,h_2$ in the $xy$-projection---in
particular, one of the quadrants above $h_2$---contains at least $n/3-1$ lines 
from $\Ltbz$. Recall that the splitting plane~$h_3$
is a semi-tilted plane parallel to the $y$-axis. Hence, $h_3$ is crossed by all lines from~$\Ltbz$,
and so there would be an octant in $\A(H)$ that intersects at least $n/3-1$ lines,
thus contradicting our assumptions.
\end{proof}
\end{myquote}
Next we restrict the possible positions for the semi-tilted splitting plane~$h_3$.
\begin{myquote} 
\vspace*{-5mm}
\begin{claim} \label{claim:3}
 $h_3$ must have at least one line of $\Ltby$ on either side of it.
\end{claim}
\begin{proof}
Observe that $h_3$ can contain at most two lines from~$\Ltby$. 
Suppose for a contradiction that there are at least $n/3-2$ lines to one side of~$h_3$,
and consider the four octants of $\A(H)$ lying to this side.
Every line from $\Ltbx$ and $\Ltbz$ intersects at least one of these octants,
except at most one line from $\Ltbx$ contained in~$h_2$ and at most one line 
from $\Ltbz$ contained in~$h_1$. Moreover, each of the at least
$n/3-2$ lines from $\Ltby$ not contained in $h_3$, intersects~$h_2$
and, hence, two of the octants. It follows that there is an octant that
intersects
\[
\ceil*{\frac{(2n/3-2) + 2(n/3-2)}{4}}
= 
\ceil*{\frac{4n/3 - 6}{4}}
= 
n/3 -1
\]
lines, thus contradicting our assumptions.
\end{proof}
\end{myquote}
Let $i$ be the number of lines of $\Ltbx$ below $h_2$. 
\begin{myquote}   
\vspace*{-5mm}
\begin{claim} \label{claim:29ths}
$i \geq \frac{2}{9}n-\frac{1}{3}$.
\end{claim}
\begin{proof}
Consider the two quadrants above $h_2$ and the four octants of $\A(H)$ corresponding to
those quadrants. Each line of $\Ltbx$ above $h_2$ intersects $h_1$ and $h_3$, and thus 
intersects three of these octants, for a total of $3(n/3-i-1)$ octant-line intersections.  
Each line of $\Ltby$, except for at most two lines contained in $h_3$, intersects one 
octant, thus contributing $n/3-2$ octant-line intersections. 
Finally, each line of $\Ltbz$, except for at most one line contained in~$h_1$,
intersects two octants (since they intersect $h_3$), contributing $2(n/3-1)$ octant-line intersections. 
Hence, there is an octant intersecting
\[
\ceil*{\frac{3(n/3-i-1) + (n/3-2)+2(n/3-1)}{4}}
=
\ceil*{n/3 - 2 + \frac{2n/3-3i+1}{4}}
\]
lines. In order not to contradict the assumption that no octant intersects more
than $n/3-2$ lines, we must thus have $2n/3-3i+1\leq 0$. Hence, $i\geq 2n/9 -1/3$, as claimed.
\end{proof}
\end{myquote}
To finish the proof of Lemma~\ref{lem:3b-lb-general}, we switch views and consider 
the projection onto the $xz$-plane. Here $h_3$ shows up as a slanted line, as shown 
in Fig.~\ref{fig:semi-tilted-1}(ii). The are two cases, depending on the slope of $h_3$
in the $xz$-projection.
\medskip

\noindent\emph{Case 1: The projection of $h_3$ onto the $xz$-plane has positive slope.}
\\[2mm]
Consider the \SW-quadrant in Fig.~\ref{fig:semi-tilted-1}(i), that is, the quadrant
below~$h_2$ and to the left of $h_1$. This quadrant corresponds to a vertical
column in $\Reals^3$, which is cut into two octants by the semi-tilted plane~$h_3$.
Clearly, each line in $\Ltby$, except for at most two lines contained in $h_3$,
intersects one of these two octants. Furthermore, each of the $i$~lines of $\Ltbx$ 
below $h_2$ intersects both octants, because $h_3$ has positive slope in the
$xz$-projection. Indeed, by Claim~\ref{claim:3} the plane~$h_3$ must have at least one
line from $\Ltby$ on either side of it---see Fig.~\ref{fig:semi-tilted-1} where these lines show
up as points---and together with the fact that $h_3$ has positive slope this implies
that all intersections of $\Ltbx$ with $h_3$ lie to the left of~$h_1$.
Thus, all $i$~lines of $\Ltbx$ below $h_2$ intersects both octants, contributing $2i$
octant-line intersections. Hence, the total number of octant-line intersections
in the two octants is at least $(n/3-2)+2i$.  Since $i \geq \frac{2}{9}n-\frac{1}{3}$
by Claim~\ref{claim:29ths} there is an octant that intersects at least
\[
\ceil*{\frac{(n/3 - 2) + 2i}{2}}
\geq
\ceil*{\frac{n/3 + (4n/9 -2/3) - 2}{2}}
=
\ceil*{n/3 - 2 + \frac{ (n/9  - 2/3)}{2}}
\]
lines. Since $n\geq 9$ there is an octant intersecting more than $n/3-2$
lines, thus giving the desired contradiction for Case~1.
\medskip

\noindent\emph{Case 2: The projection of $h_3$ onto the $xz$-plane has negative slope.}
\\[2mm]
\begin{figure}
\begin{center}
\includegraphics{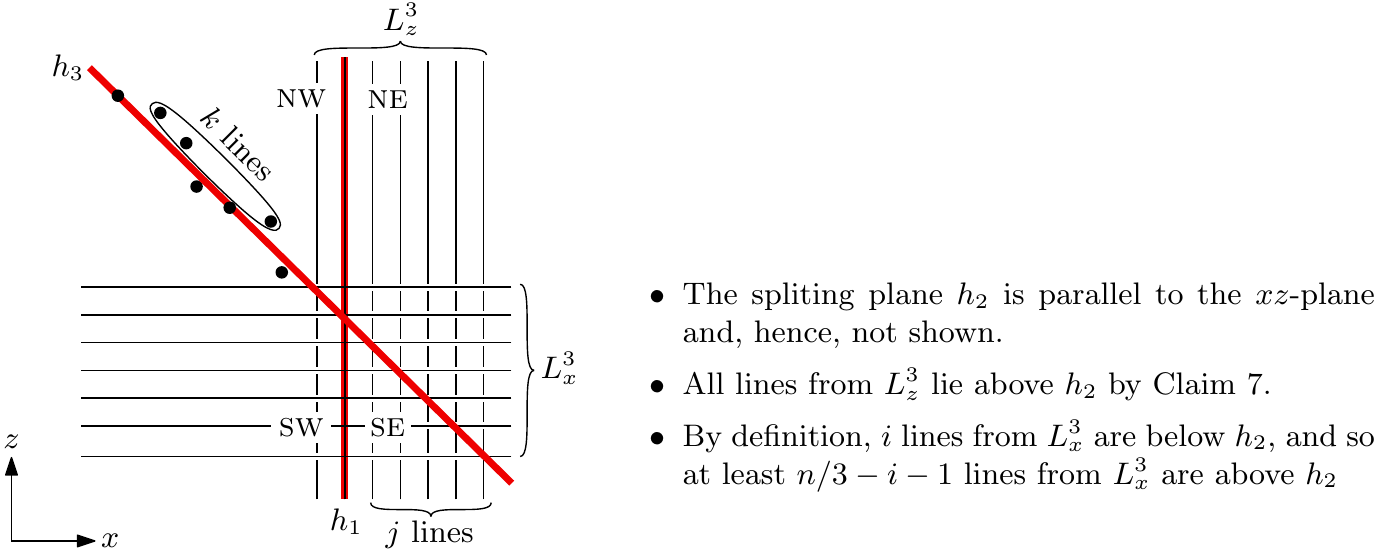}
\end{center}
\caption{Illustration for Case~2 of the proof of Lemma~\ref{lem:3b-lb-general}.
Recall that the lines in the construction were perturbed slightly. For the lines in
$\Ltby$, which show up as points in the figure, these perturbations are shown 
larger than they actually are.
Due to the perturbations, $h_3$ can contain at most two lines from $\Ltby$.
}
\label{fig:semi-tilted-2}
\end{figure}
This case is illustrated in Fig.~\ref{fig:semi-tilted-2}. 
Let $j$ be the number of lines of $\Ltbz$ to the right of~$h_1$
and let $k$ be the number of lines of $\Ltby$ above~$h_3$; 
see Fig.~\ref{fig:semi-tilted-2}.   We label the ``quadrants'' induced by $h_1,h_3$ 
the $xz$-projection as \NE, \SE, \SW, \NW.
Each of the these quadrants defines a column 
in~$\Reals^3$, which is partitioned into two octants by~$h_2$.
We label the octants in the \NE-column that are above and below $h_2$
by \NEt and \NEb, respectively. The octants in the other columns are labeled similarly.
\begin{itemize}
\item Consider the \NWt octant. It meets at least $n/3-j-1$ lines of $\Ltbz$, because
    all lines of $\Ltbz$ lie above~$h_2$ by Claim~\ref{claim:1}. Moreover, it meets
    $k$ lines of~$\Ltby$. (It may also meet lines from $\Ltbx$ but we need not take 
    them into account.) Since any octant is assumed to intersect less
    than $n/3-1$ lines, we must have $(n/3-j+1)+k < n/3-1$. Hence, $j>k+2$.
\item Now consider the \SWb octant. It meets $i$ lines from $\Ltbx$ and at least
    $n/3-k-2$ lines of $\Ltby$, and so we must have $i + (n/3-k-2) < n/3-1$.
    Hence, $k>i-1$.
\item Finally, consider the \NEt octant. It meets $j$ lines from $\Ltbz$, because
    all lines of $\Ltbz$ lie above~$h_2$ by Claim~\ref{claim:1}. Note that
    all lines from $\Ltbx$ intersect the \NE quadrant, since $h_3$ has negative
    slope in the $xz$-projection. At least $n/3-i-1$ of these lines
    lie above $h_2$ and, hence, intersect the \NEt octant.
    We can conclude that we must have $j + (n/3-i-1) < n/3-1$, and so~$i>j$.
\end{itemize}
Putting these three inequalities together we obtain
\[
j > k+2 > i+1 > j+1.
\]
This is the desired contradiction and finishes the proof.
\end{proof}

\section{Better lower bounds for axis-parallel splitting planes}
\label{sec:lower-bounds}
In this section we give improved lower bounds on $\gperp(n)$ and $\gpar(n)$.
To this end we present a family of configurations $\Lfour(n)$ of $n$ lines, with $n$ divisible by 8, 
and prove that any decomposition of the type under consideration must have a cell intersected 
by at least $3n/8 - 1$~lines. Interestingly, for mutually orthogonal splitting planes the $-1$ term disappears.
\begin{theorem}
  \label{th:unbalancedlb}
  For every $n$ divisible by 8, there exists a configuration $\Lfour=\Lfour(n)$ such that 
  $\gperp(\Lfour(n)) = 3n/8$ and $\gpar(\Lfour(n)) = 3n/8 - 1$. 
  Hence, for any $n\geq8$, we have $\gperp(n) \geq \frac{3}{2}\floor{n/4}$ and $\gpar(n) \geq \frac{3}{2}\floor{n/4}-1$.
\end{theorem}
The rest of this section is devoted to the proof of Theorem~\ref{th:unbalancedlb},
where $\Lfour$ is defined as the union of the following sets:
\begin{align*}
  \Lfour_x & \mydef \{ (0,i,i + n/4) + \lambda (1,0,0): i \in [n/4, n/2) \},\\
  \Lfour_y & \mydef \{ (i+n/4,0,i) + \lambda (0,1,0):  i\in [0, n/4) \}, \text{ and} \\
  \Lfour_z & \mydef \left\{ (i,i,0) + \lambda (0,0,1): i \in [0, n/4) \cup [n/2, 3n/4) \right\}.
\end{align*}
Note that $|\Lfour_x| = |\Lfour_y| = n/4$ and $|\Lfour_z| = n/2$, and that, up to symmetries, 
any projection to an axis-parallel plane looks like one of projections in Fig.~\ref{fi:unbalancedlb}.
\begin{figure}[ht]
\begin{center}
\includegraphics[scale=0.95]{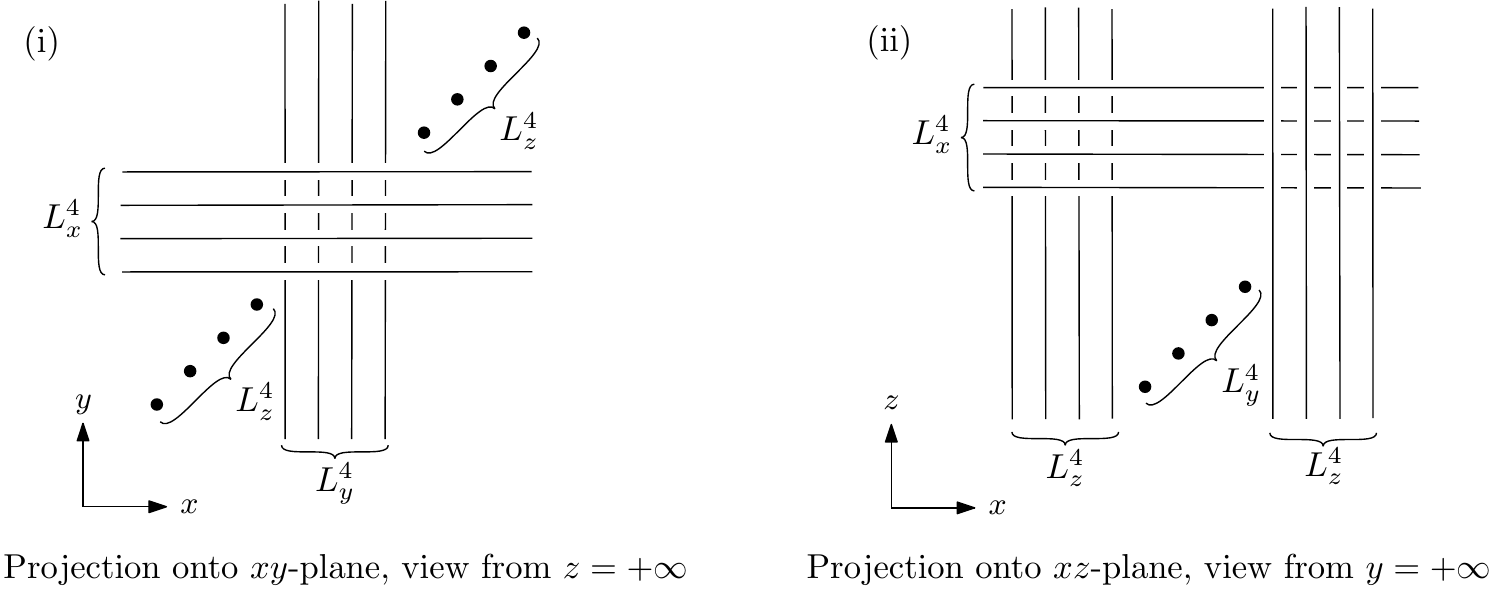}
\end{center}
\caption{Line set $\Lfour$ from the proof of Theorem~\ref{th:unbalancedlb}}
\label{fi:unbalancedlb}
\end{figure}
We will prove the theorem in two steps, by first considering the case of mutually
orthogonal splitting planes and then the case of arbitrary axis-parallel splitting planes.
In the remainder of this section, we assume that $n$ is divisible by~8 and that~$n\geq 8$.
\begin{lemma}\label{lem:4b-lb-perp}
For any $n$ divisible by $8$, we have $\gperp(\Lfour(n)) = 3n/8$.
\end{lemma}
\begin{proof}
Consider three splitting planes $h_1$, $h_2$, and $h_3$, orthogonal to the $x$-, $y$-, 
and $z$-axis, respectively. Let $\ell_1\colon x = i$ and $\ell_2\colon y = j$ be the 
projections of $h_1$ and $h_2$ onto the $xy$-plane; see Fig.~\ref{fi:unbalancedlb-1}.
Suppose that $h_3$ has $k$ of the lines from $\Lfour_x \cup \Lfour_y$ lying strictly above it. 
(Note that for the purposes of lower-bound analysis we can assume that $h_1$, $h_2$, and $h_3$ 
each contain a line of $\Lfour$, as shifting them until they do can only decrease the number of 
lines of $\Lfour$ meeting each open octant.  We will make this assumption hereafter in this proof, 
that is, we suppose that $i$, $j$, and $k$ are integers in the relevant range.)
We use $\NE$ to denote the open north-east quadrant defined by $\ell_1,\ell_2$ and define $\SE$, $\SW$, and $\NW$ similarly. 
For a quadrant $\mbox{\sc q} \in \{\NW, \NE, \SW, \SE\}$, we define 
\mbox{\sc q-top} and \mbox{\sc q-bottom} to be the open octants induced by $h_3$.
With a slight abuse of notation, we also use $\NE$ (and, similarly,
the other variables) to denote the total number of objects
(lines and points) incident to the region~$\NE$.

We first observe that by setting $i = j = 3n/8$ and $k = n/4$ we 
obtain~$\gperp(\Lfour) \leq 3n/8$. It remains to argue that $\gperp(\Lfour) \geq 3n/8$.
Up to translation,  rotation, and reflection, it suffices to consider the following cases.
\medskip

\begin{figure}[b]
  \begin{center}
    \includegraphics[width=1.0\textwidth]{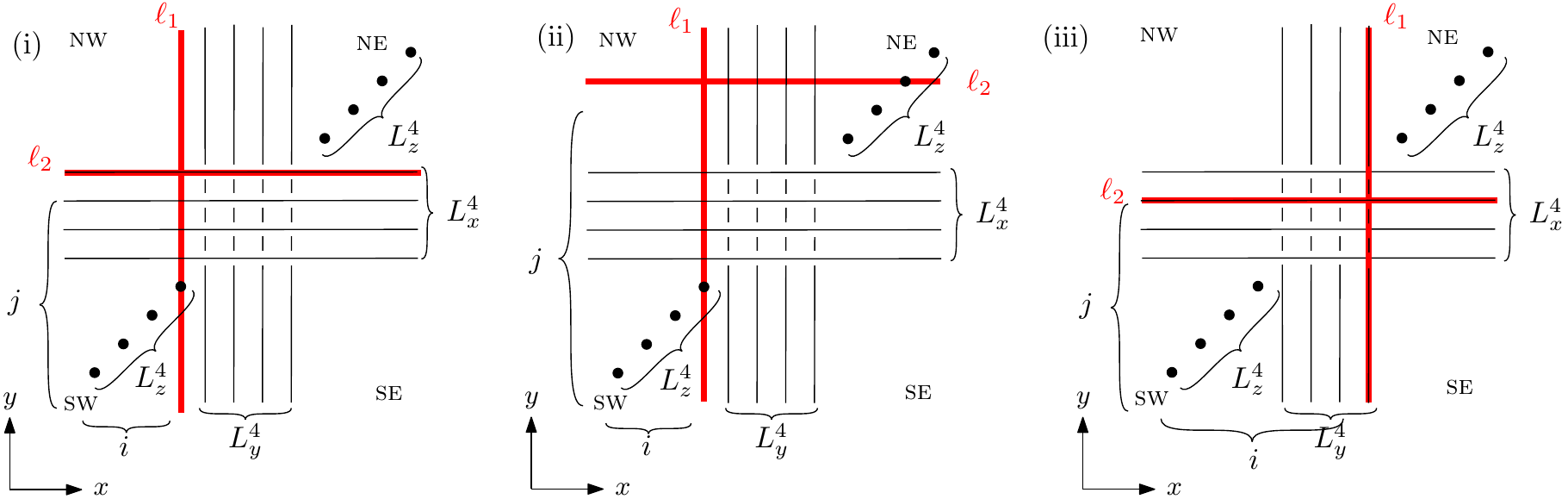}
  \end{center}
  \caption{Different cases of partitioning with three orthogonal planes, projected to the $xy$-plane showing view from $z=+\infty$.} 
  \label{fi:unbalancedlb-1}
\end{figure}

\noindent\emph{Case 1: $0 \leq i < n/4$ and $0 \leq j < n/2$; 
see Fig.~\ref{fi:unbalancedlb-1}(i).}
\\[2mm]
Observe that $\NE$ is incident to, at least, all $n/4$ lines in $\Lfour_y$, 
and $n/4$ lines of $\Lfour_z$. It follows that the number of incidences in $\NEt$ 
and $\NEb$ together is at least $n/4 + 2(n/4) - 1 = 3n/4 - 1$, 
implying that at least one of $\NEt, \NEb$ is incident to 
$$\left\lceil\frac{3n}{8} - \frac{1}{2}\right\rceil\geq \frac{3n}{8}$$ lines, where the inequality uses the fact that since $n$ is divisible by 8.
\medskip

\noindent\emph{Case 2: $0 \leq i < n/4$ and $n/2 \leq j < 3n/4$; 
see Fig.~\ref{fi:unbalancedlb-1}(ii).}
\\[2mm]
We first additionally assume that $0 \leq k < n/4$. Thus, $h_3$ meets some line of $\Lfour_x$. 
$\NEb$ meets $n/4$ lines of $\Lfour_y$ and $3n/4 - j - 1$ lines of $\Lfour_z$, 
so $n - j - 1$ lines in total. If $j < 5n/8$ then we have $n - j - 1 > 3n/8 - 1$ 
and, hence, $n - j - 1 \geq 3n/8$.
So we may assume $j \geq 5n/8$. Observe that $\SE$ is incident 
to all $n/2$ lines in $\Lfour_x \cup \Lfour_y$, and at least $j - n/2 \geq n/8$ lines of $\Lfour_z$.
Since $h_3$ contains a line of $\Lfour_x$, the number of incidences in $\SEt$ and $\SEb$ is 
at least $2(n/8) + n/2 - 1 = 3n/4 - 1$, implying that at least one of $\SEt$ 
and $\SEb$ is incident to at least $$ \left\lceil\frac{3n}{8} - \frac{1}{2}\right\rceil = \frac{3n}{8}$$ lines.

Now suppose $n/4 \leq k < n/2$, i.e., $h_3$ meets some line of $\Lfour_y$.  
$\SWt$ meets $i$ lines of $\Lfour_z$ and all $n/4$ lines of $\Lfour_x$, 
i.e., $\SWt = i + n/4$. When $i \ge n/8$ (recall that $n$ is a multiple of~$8$),  
we get $\SWt \ge 3n/8$, so we can assume $i < n/8$. Observe that $\SE$ is incident 
to all $n/2$ lines in $\Lfour_x \cup \Lfour_y$, and 
$$\frac{n}{4} - i-1 > \frac{n}{4} - \frac{n}{8} - 1 = \frac{n}{8} - 1,$$ that is, 
at least $n/8$ lines of $\Lfour_z$. Thus the number of incidences in $\SEt$ and
$\SEb$ together is at least $$2\left(\frac{n}{8}\right) + \frac{n}{2} - 1 = \frac{3n}{4} - 1,$$ implying 
that at least one of $\SEt$ and $\SEb$ meets at least $3n/8$ lines.
\medskip

\noindent\emph{Case 3:  $n/4 \leq i < n/2$ and $n/4 \leq j < n/2$; 
see Fig.~\ref{fi:unbalancedlb-1}(iii).}
\\[2mm]
Suppose also that $0 \leq k < n/4$ (i.e., $h_3$ contains a line of $\Lfour_x$). 
Then, all but one of the lines of $\Lfour_y$ (so $n/4 - 1$) are incident to 
either $\NEb$ or $\SWb$. Additionally, $\NEb$ and $\SWb$ are each incident to $n/2$ 
lines of $\Lfour_z$. That is, number of incidences in $\NEb$ and 
$\SWb$ is at least $n/4 - 1 + n/2 = 3n/4 - 1$. 
It follows that at least one of these cells is incident to $3n/8$ lines.

Now suppose $n/4 \leq k < n/2$. i.e., $h_3$ contains a line of $\Lfour_y$.
Then, all but one of the lines of~$\Lfour_x$ (so $n/4 - 1$) are incident to 
either $\NEt$ or $\SWt$. Additionally, $\NEt$ and $\SWt$ are each incident to $n/2$ lines of $\Lfour_z$.
Then the number of incidences in $\NEt$ and $\SWt$ is at least $n/4 - 1 + n/2 = 3n/4 - 1$. 
It follows that at least one of these cells is incident to $3n/8$ lines.
\medskip

 We conclude that all three cases give the desired number of incidences, which
 finishes the proof of the lemma.
\end{proof}
To finish the proof of Theorem~\ref{th:unbalancedlb} it remains to deal with
axis-parallel splitting planes that need not be mutually orthogonal.
\begin{lemma}\label{lem:4b-lb-par}
For any $n$ divisible by~8, we have $\gpar(\Lfour(n)) = 3n/8 -1$.
\end{lemma}
\begin{proof}
We first observe that the case of all three planes being orthogonal to the same axis 
is not interesting: Such planes would cut each line of one of the sets $\Lfour_x,\Lfour_y,\Lfour_z$ 
into four pieces, producing a total of at least 
$$4\cdot\frac{n}{4}+\frac{n}{4}-1+\frac{n}{2}-1=\frac{7n}{4}-2$$ line-cell incidences, 
so at least one of the four cells would meet at least $$\frac{1}{4}\left(\frac{7n}{4}-2\right)=\frac{7n}{16}-\frac12>\frac{3n}{8}-1$$ lines.
Hence, we will focus on partitions where two of the planes are parallel. Up to a 
permutation and reorientation of the coordinates, it is sufficient to consider the 
three cases illustrated in Figure~\ref{fi:unbalancedlb-2}. 
As before, it is safe to assume that 
each of the three planes contains a line of~$\Lfour$. 

We need prove that in each of the three cases, there is a cell that meets at least $3n/8 - 1$ lines,
and that this bound can also be achieved as an upper bound in at least one of the cases.
\begin{figure}
  \begin{center}
    \includegraphics[width=1.0\textwidth]{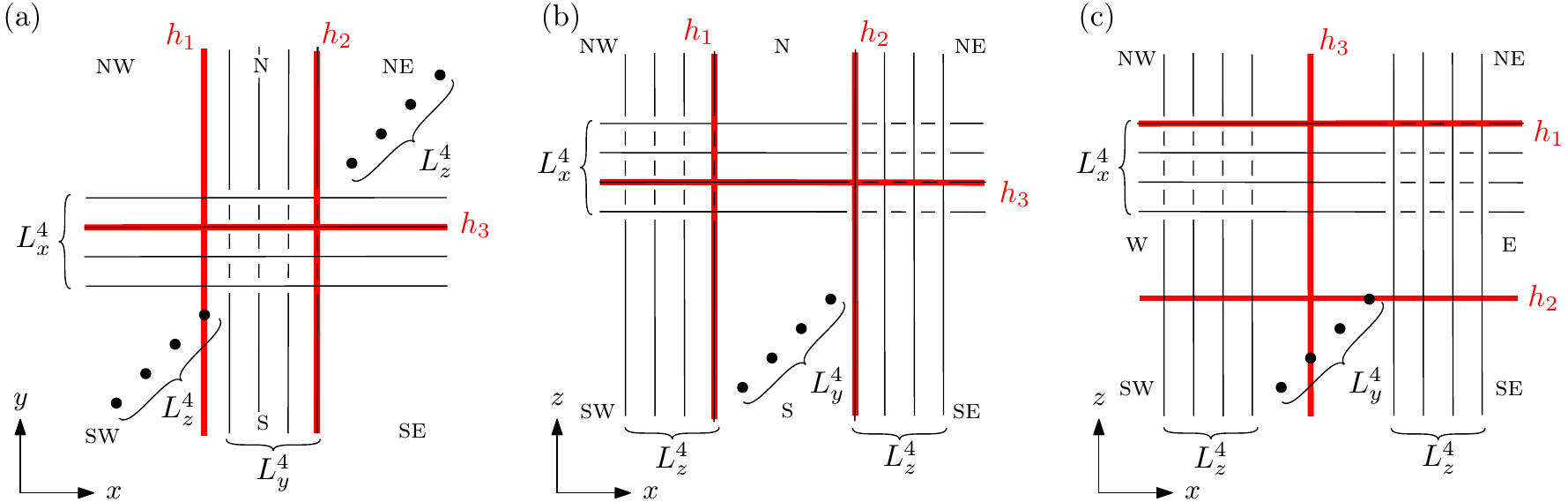}
  \end{center}
  \caption{Different cases in the proof of Lemma\protect\ref{lem:4b-lb-par}. 
           For Case~(a) the projection onto the $xy$-plane is shown, viewed from $z=+\infty$. 
           For Cases~(b)~and~(c) the projection onto the $xz$-plane is shown, viewed from $y=+\infty$.}
  \label{fi:unbalancedlb-2}
\end{figure}
\medskip

\noindent \emph{Case (a)\textnormal{:} $h_1$ and $h_2$ are parallel to the $yz$-plane and $h_3$ is parallel to the $xz$-plane.}
\\[2mm]
We label the cells as \NW, \north, \NE, \SW, \south, and \SE. Suppose there are $i$~lines 
from~$\Lfour_y \cup \Lfour_z$ strictly to the left of~$h_1$, and $j$ lines from $\Lfour_y \cup \Lfour_z$ 
strictly to the right of $h_2$, and $k$ lines from $\Lfour_x \cup \Lfour_z$ strictly above $h_3$. 
Due to the symmetry in the configuration, we may assume that the number of lines from
$\Lfour_x$ below~$h_3$ is at least the number of lines from $\Lfour_x$ above~$h_3$.
We can assume that $h_3$ contains a line of $\Lfour_x$ (and not a line of $\Lfour_z$), 
since otherwise $$\SW + \south + \SE = 3\left(\frac{n}{4}\right) + \frac{n}{4} + \frac{n}{4} - 2 = \frac{5n}{4} - 2,$$ 
implying that at least one of these cells contains $$\left\lceil\frac{5n}{12} - \frac{2}{3}\right\rceil \geq \frac{3}{8}n - 1$$ lines.
Hence, from now on we assume~$0\leq k < 3n/8$. 

Note that $\SW$ is incident to $n/2 - k -1 \geq n/8$ lines of $\Lfour_x$
and $i$ lines of $\Lfour_y\cup \Lfour_z$. If $i \geq n/4$, then this implies 
$\SW \geq n/8 + i \ge 3n/8$ and so we are done. 
Similarly, if $j < n/4$, then $\south$ is incident to at least $n/8$ lines of $L_x$
and $n/4$ lines of $L_y$, for a total of $3n.8$ incidences, and we are done as well.
Hence, from here on we assume that $i < n/4$ (so $h_1$ contains a line of $\Lfour_z$),
and that $j \geq n/4$ (so $h_2$ contains a line of $\Lfour_y$). 
Observe that $\SW, \NE$ are incident to $i + j$ lines of $\Lfour_y\cup \Lfour_z$, and $n/4-1$ lines of $\Lfour_x$,
that is,  $$\SW + \NE = i + j + \frac{n}{4}-1. $$ 
Moreover, $\south$ is incident to at least $n/8$ lines of $\Lfour_x$,
to $n/2 - j-1$ lines of $\Lfour_y$, and to $n/4 - i-1$ lines of $\Lfour_z$ 
giving $$\south = \frac{n}{8} + \left(\frac{n}{2} - j - 1\right) + \left(\frac{n}{4}  - i -1\right)= \frac{7n}{8} - (i + j) - 2.$$ 
By combining these two inequalities we see that $$\SW + \NE + \south = \frac{9n}{8} - 3,$$
implying one of the three cells is incident to at least $3n/8 - 1$ lines.

This finishes the lower-bound proof for Case~(a). Note that we can also achieve an upper bound
of $3n/8 - 1$, by setting $k = 3n/8 - 1$, $i = n/4 -1$, and $j = n/4$. 
\medskip

\noindent \emph{Case (b)\textnormal{:} $h_1$ and $h_2$ are parallel to the $yz$-plane and 
      $h_3$ is parallel to the $xy$-plane.}
\\[2mm]
As before, we label the cells as \NW, \north, \NE, \SW, \south, and \SE. Suppose there 
are $i$ lines from~$\Lfour_y \cup \Lfour_z$ strictly to the left of $h_1$, and 
$j$~lines from $\Lfour_y \cup \Lfour_z$ strictly to the right of $h_2$, 
and $k$ lines from $\Lfour_x \cup \Lfour_y$ strictly above $h_3$.

Each of the lines in $\Lfour_y$ and $\Lfour_z$ are incident to exactly one of $\SW, \south, \SE$, 
except for the two lines contained in $h_1$ and in $h_2$. If $k < n/8$, then 
at least $n/8$ of the lines in $\Lfour_x$ are incident to each of the 
cells $\SW, \south, \SE$, implying  
$$\SW + \south + \SE = 3\left(\frac{n}{8}\right) + \frac{3n}{4} - 2 = \frac{9n}{8} - 2.$$ 
Hence, one of $\SW, \south, \SE$ is incident to at least $\ceil{3n/8 - 2/3} = 3n/8$ lines,
and we are done. From here on, we therefore assume $k \geq n/8$. 

Each line in $\Lfour_x$ is incident to either of $\SW$ or $\NW$, except for the line 
contained in $h_3$. If $i \geq n/4$ then $n/4$ lines in $\Lfour_z$ 
are incident to both $\SW$ and $\NW$, and so  
$$\SW + \NW = 2\left(\frac{n}{4}\right) + \frac{n}{4} - 1 = \frac{3n}{4} - 1.$$
It follows that $\SW$ or $\NW$ must be incident to at least $3n/8$ lines,
and we are done again. So we can assume $i < n/4$. By symmetry, we also can 
assume $j < n/4$.

Each of the $k$ lines of $\Lfour_x$ above $h_3$ are incident to both $\NW$ and $\NE$.
The cells $\NW$ and $\NE$ are also incident to $i$ and $j$ lines from $\Lfour_z$, respectively, 
and so $\NW + \NE = 2k + i + j$. Combining this with $\NW + \NE \leq 3n/4 - 2$, 
we get that $i + j \leq 3n/4 - 2k - 2$. On the other hand, $\south$ is incident 
to $n/2 - k - 1$ lines of $\Lfour_x \cup \Lfour_y$ and $n/2 - (i+j) - 2$ lines of $\Lfour_z$.
Hence,
\begin{align*}
\south & =
    \left(\frac{n}{2} - k - 1\right) + \left(\frac{n}{2} - (i+ j)-2\right) \\
    & = n - k - (i + j) -3 \\
    & \geq n - k - \left(\frac{6n}{8} - 2k - 2\right) - 3 \\
    & =  \frac{n}{4} + k - 1 \\
    & \geq \frac{3n}{8} - 1.
\end{align*}
This finishes the proof of the lower bound for Case~(b). Note that also for this case
we can achieve an upper bound of $3n/8 - 1$ lines, by setting $k = n/8$ and $i = j = n/4 - 1$.
\medskip

\noindent \emph{Case (c)\textnormal{:} $h_1$ and $h_2$ are parallel to the $xy$-plane 
      and $h_3$ is parallel to the $yz$-plane.}
\\[2mm]
Label the cells as $\NW, \west, \SW, \NE, \east$, and $\SE$. Let $k$ be the number of lines 
in $\Lfour_y \cup \Lfour_z$ that are strictly to the left of~$h_3$. By symmetry,
we may assume that $k \geq 3n/8$. Now each of the cells $\NW, \west, \SW$ are 
incident to~$n/4$ of the lines from $\Lfour_z$. All $n/4$ lines in $\Lfour_x$ 
and at least $n/8$ lines in $\Lfour_y$ are incident to exactly one of these cells, 
except for two that might be contained in $h_1$ and $h_2$. Hence, 
$$\NW + \west + \SW \geq 3\left(\frac{1}{4}n\right) + \frac{n}{8} + \frac{n}{4} - 2 = \frac{9n}{8} - 2,$$ 
implying one of these cells is incident to at least $\ceil{3n/8 - 2/3} = 3n/8$ lines.

This finishes the lower-bound proof for Case~(c) and, hence, the proof of the lemma.
\end{proof}

\section{Upper bounds}
\label{sec:upper-bounds}
We now prove upper bounds on $\gperp(n)$ and $\gpar(n)$.
More precisely, we present algorithms that produce, for any set $L$ of $n$ lines,
a decomposition of the type under consideration, with each cell intersected 
by at most a certain number of the lines.

Let $L := L_x \cup L_y\cup L_z$ be a set of axis-parallel lines in $\Reals^3$,
where $L_x,L_y,L_z$ denote the subsets parallel to the $x$-, $y$-, and $z$-axis, respectively. We say $L_x$ is in \emph{general position} if no pair of distinct lines in $L_x$ share $y$- or $z$-coordinates; we define general position for $L_y$ and $L_z$ similarly. We say $L$ is in general position if $L_x$, $L_y$, and $L_z$ are in general position, and no two lines in~$L$ have a common point. We first argue that, for the purposes of upper bounds, it suffices to restrict our attention to sets in general position.

\begin{lemma}
\label{lem:generalposition}
Let $L = L_x \cup L_y\cup L_z$ be a set of axis-parallel lines not in general position. Then there exists a set $L' = L'_x \cup L'_y\cup L'_z$ of axis-parallel lines in general position such that
\[\gpar(L)\leq \gpar(L') \quad \mbox{ and } \quad \gperp(L)\leq \gperp(L'). \]
\end{lemma}
\begin{proof}
Suppose that $L = L_x \cup L_y\cup L_z$ is a set of lines not in general position. Assume, without loss of generality, that the intersection points of lines in $L_x$ (resp. $L_y$, and $L_z$) with the plane $x = 0$ (resp. $y = 0$, and $z = 0$) have integer coordinates. We obtain $L'$ by a generic perturbation of the lines of $L$. More specifically, for each line $\ell \in L$, we let $\ell'$ be a generic line parallel to $\ell$ inside a tube of radius $1/3$ centered at $\ell$.

Consider a triple $H' = (h'_1, h'_2, h'_3)$ of axis-parallel splitting planes (which need not be mutually orthogonal). Suppose, without loss of generality, that $h'_1$ is orthogonal to the $x$-axis and is given by $x=\alpha'$ for some $\alpha' \in \Reals$. Let $h_1$ be the plane given by $x = \alpha$ where $\alpha$ is $\alpha'$ rounded to the nearest integer with ties broken arbitrarily; and define $h_2, h_3$ similarly. Let $H = (h_1, h_2, h_3)$ be the resulting triple of axis-parallel splitting planes. 

Let  $\mathcal{C}' \in \mathcal{A}(H')$, and let $\mathcal{C} \in \mathcal{A}(H)$ be the corresponding cell. (Since two or more planes in $H'$ may be ``rounded'' to the same plane in $H$, the cell $\mathcal{C}$ can be the empty set. In this case the following claim trivially holds.)
We claim that the number of incidences of $\mathcal{C}$ with lines of $L$ is at most the number of incidences of $\mathcal{C}'$ with lines of $L'$. 
This claim follows from the observation that, for each plane $h_i \in H$, the number of lines of $L$ lying strictly on one side of $h_i$ is upper bounded by the number of lines of $L'$ lying on the same side of the corresponding $h'_i \in H$. To see this, note that if $\ell' \in L'$ lies on $h_i$, on $h'_i$, or between them, then the corresponding line $\ell \in L$ lies on~$h_i$. 
\end{proof}

\subsection{Upper bounds on \texorpdfstring{$\gpar$}{gpar}}
We start with a simple observation.
\begin{observation}\label{obs:good-label}
If $\max\left(|L_x|,|L_y|,|L_z|\right) = m$ then 
there is a set $H$ of three axis-parallel planes (two of which are parallel) such that any cell in $\A(H)$
meets at most $(5n-m)/12$~lines from $L$. 
\end{observation}
\begin{proof}
  Assume without loss of generality that $L_z$ is the smallest of the three sets.  By assumption, $|L_z|\leq \floor*{(n-m)/2}$.
  Partition $L_x\cup L_y$ into three equal-size subsets using two planes $h_1,h_2$ parallel to the
  $xy$-plane, and partition $L_z$ into two equal-size subsets using a plane $h_3$ parallel to the $yz$- or $xz$-pane.  As in earlier arguments, we can always choose the planes so that they contain a line of~$L$.
  Set $H := \{ h_1,h_2,h_3\}$.
  Then the number of lines each of the six cells in $\A(H)$ meets is at most
  \begin{align*}
    \ceil*{\frac{|L_x|+|L_y|-2}{3}} + \ceil*{\frac{|L_z|-1}{2}} & \leq
    \frac{|L_x|+|L_y|}{3} + \frac{|L_z|}{2} \\
    & = \frac{n - |L_z|}{3} + \frac{|L_z|}{2}\\
    & \leq \frac{n}{3} + \frac{\floor*{(n-m)/2)}}{6}\\
    & \leq \frac{5n-m}{12}. \qedhere
  \end{align*}
\end{proof}
The following theorem gives an upper bound on~$\gpar$.
\begin{theorem}\label{th:7nover18}
For any set $L$ of $n$ axis-parallel lines in $\Reals^3$, there is
a set $H$ of three axis-parallel planes (two of which are parallel) such that any cell in $\A(H)$
meets at most $7n/18$ lines from $L$. Hence, $\gpar(n) \leq \floor{7n/18}$. 
\end{theorem}
\begin{proof}
Define $L_x$, $L_y$, and $L_z$ as above.
The largest of the groups has size at least $\ceil{n/3}$.  The theorem now follows from Observation~\ref{obs:good-label} with $m=\ceil{n/3}$. 
\end{proof}

\subsection{Upper bounds on \texorpdfstring{$\gperp$}{gperp}}

\begin{theorem}
\label{thm:512upper-orig}
For any set $L$ of $n$ axis-parallel lines, there is
a set $H$ of three planes, one
orthogonal to each axis direction, such that any cell in $\A(H)$
meets at most $\ceil{5n/12}$ of the lines; in other words, $\gperp(n)\leq \ceil{5n/12}$.
\end{theorem}
\begin{proof}
Assume without loss of generality that $|L_z|\geq \max(|L_x|, |L_y|)$, and 
consider the projection of $L$ onto the $xy$-plane. In the projection the 
lines from $L_z$ become points, the lines from the other sets are still lines.
We denote the splitting planes orthogonal to the
$x$-, $y$-, and $z$-axis by $h_1$, $h_2$, and $h_3$, respectively.
We will first explain how we pick $h_1$ and $h_2$---in the projection
these correspond to splitting lines, which we denote by $\ell_1$ and 
$\ell_2$, respectively---and then finish the construction by placing~$h_3$.
We distinguish two cases.
\medskip

\noindent\emph{Case 1: $|L_z| \geq n/2$}.
\\[2mm]
  This is the easy case: we pick $\ell_1$ and $\ell_2$ such that each
  open quadrant contains at most $\floor{|L_z|/3}$ points, and we pick $h_3$ such that
  at most half the lines from $L_x\cup L_y$ are above $h_1$ and at most half are below.
  Thus each cell in the resulting decomposition intersects at most
  \[
  \floor*{\frac{|L_z|}{3}} + \floor*{\frac{|L_x|+|L_y|}{2}}
  \leq \frac{|L_x|+|L_y|+|L_z|}{3} + \frac{|L_x|+|L_y|}{6}    
  \leq \frac{5n}{12},
  \]
  lines, since $|L_x|+|L_y|+|L_z|=n$ and $|L_x|+|L_y| \leq n/2$.
\medskip

\noindent\emph{Case 2: $|L_z| < n/2$}.
\\[2mm]
  For two given splitting lines $\ell_1$ and $\ell_2$ in the $xy$-plane, we use
  $\NE$ to denote the number of lines in $L$ whose projection 
  (which can be a line or a point) intersects
  the open north-east quadrant defined by $\ell_1,\ell_2$,
  and we define $\SE$, $\SW$, and $\NW$ similarly. Let 
  $\north = \NW + \NE$, and define $\east, \south, \west$ similarly.
  Note that lines from $L_x$ that lie in the northern part are
  counted twice in $\north$, once for their intersection with the north-west quadrant 
  and once for their intersection with the north-east quadrant.
  Finally, we use
  $\NE_x$ to denote the number of lines from $L_x$ intersecting the north-east quadrant, 
  and we use $\NE_y$, $\south_z$, and so on, in a similar way.   Finally, let $\total = \NE + \NW + \SE + \SW$ denote the total number of incidences. 

Let $W(L) := 2(|L_x| + |L_y|) + |L_z| = n + |L_x| + |L_y|$, and note that
\[ \frac{3n+1}{2} \leq W(L) \leq \floor*{5n/3} \]
which follows from the facts that $\max(|L_x|, |L_y|) \leq |L_z|$ and $|L_z|<n/2$.
We will require $\ell_1$ and $\ell_2$ together each contain the projection of a line; note that such a line may subtract 1 from the total count~$\total$ (if its projection is a point), or 2 (if its
projection is a line). Hence,  
  \[
  \frac{3n - 7}{2} \leq W(L) - 4 \leq \total \leq W(L) - 2 \leq \floor{5n/3} - 2.
  \]
  We now explain how to pick $\ell_2$ (and, hence, $h_2$), the splitting line orthogonal 
  to the $y$-axis. Place $\ell_2$ at
  the highest $y$-coordinate where we still have $\south \leq \frac{\floor{5n/3}}{2}$. This is always possible since $\frac{\floor{5n/3}}{2} < \frac{3n - 7}{2}$ for all $n \geq 2$. 
  Now, we have 
  \[
  \frac{\floor{5n/3}}{2} - 1 \leq \south \leq \frac{\floor{5n/3}}{2}
  \]
  where the lower bound comes from the fact that moving $\ell_2$ could change the number of incidences by two (if $\ell_2$ contains a line of $L_y$). Furthermore, since $\south + \north = \total \leq \floor{5n/3}-2$, we have
   \[
   \north \leq \frac{\floor{5n/3}}{2}-1.
   \]
 We also assume, without loss of generality, that $\south_z \geq \north_z$. Indeed, if this is not the case, we can interchange 
  the roles of $\south$ and $\north$; simply enlarge $\north$ (by shifting $\ell_2$ down) 
  until the above inequalities hold in interchanged form.
  
  Now pick $\ell_1$ (and, hence, $h_1$) such that 
  $\floor{\south/2}\leq \SW\leq \ceil{\south/2}$ and
  $\floor{\south/2}\leq \SE\leq \ceil{\south/2}$.
  From now on we assume, without
  loss of generality, that $\NE \geq \NW$. Note that this implies that
  $\NW\leq \floor{\north/2}$.  
  Furthermore, we have 
  \[\SW,\SE \leq \ceil{\south/2} \leq  \ceil{\floor{5n/3}/4} \leq \ceil{5n/12}
  \]
   and  
   \[
   \NW\leq \floor{\north/2}\leq \ceil{(\floor{5n/3}-2)/4}<\floor{5n/12}.
   \] 
   Thus the corresponding
   ``columns'' in $\Reals^3$ already have the desired number of incidences.

   We choose the remaining splitting plane~$h_3$
   such that, within the north-eastern column, at most half the lines from
   $L_x\cup L_y$ are above $h_3$ and at most half are below (and one line is contained in $h_3$).
   We conclude that the number of lines intersected by each of the two cells 
   resulting from splitting this column by $h_3$ is at most
   \begin{equation}
   \label{eq:nebound}
   \floor*{\frac{\NE_x+\NE_y}{2}} +\NE_z  \leq \frac{\NE+\NE_z}{2}.
   \end{equation}
   To bound the expression in \eqref{eq:nebound}, we rely on the following.
   
  \begin{claiminproof} \label{claim:NE-orig}
  With $\ell_1, \ell_2$ as above, we have  $\NE \leq n - \SW$.
  \end{claiminproof}
  \begin{proofinproof}
  Since $\north_x = 2\,\NE_x$ we have
  \[
  \north_z = \north - \north_x - \north_y 
    = \north - 2 \, \NE_x - \north_y.
  \]
  Trivially, we also have
  \[
  \NE = \NE_x + \NE_y + \NE_z \leq \NE_x + \NE_y + \north_z.
  \]
  Combining these we get
  \begin{equation}
      \NE \leq 
        \NE_x+\NE_y + (\north - 2 \,\NE_x - \north_y) = 
        \north - \north_y + \NE_y -\NE_x.    \label{eq:NE-orig}
  \end{equation}
  Moreover, since $\SE_x = \east_x - \NE_x$ and $\NE_y=\SE_y$ we have
  \[
  \SE = \SE_x + \SE_y + \SE_z = (\east_x - \NE_x) + \NE_y + \SE_z,
  \] 
  which can be rewritten as
  \begin{equation}
     \NE_y - \NE_x =  \SE - \east_x - \SE_z.   \label{eq:NExplusy-orig}
  \end{equation}
   Note also that
   \begin{equation}
     \north = \total - \south = 2 \,\north_y + 2 \,\east_x + \total_z - \south.   \label{eq:north-orig}
  \end{equation}
  Using that $\east_x\leq |L_x|$ and $\north_y\leq |L_y|$ and $\total_z\leq |L_z|$, we obtain 
  \begin{align*}
  \NE & \leq \north - \north_y + (\NE_y -\NE_x)  &  \mbox{ by (\ref{eq:NE-orig})} \\
      & = \north - \north_y + (\SE - \east_x - \SE_z)  &  \mbox{ by (\ref{eq:NExplusy-orig})} \\
      & = (2\, \north_y + 2 \,\east_x + \total_z - \south) - \north_y + (\SE - \east_x - \SE_z)  &  \mbox{ by (\ref{eq:north-orig})} \\      & \leq |L_x| + |L_y| + |L_z|  - \south + \SE - \SE_z & \\
      & = n  - \SW  - \SE_z \\
      & \leq n  - \SW, 
 \end{align*}
 which finishes the proof of the claim.
  \end{proofinproof}

   Now recall that
   \begin{equation}
   \SW \ge \floor{\south/2} \ge \floor*{\frac{\floor{5n/3} - 2}{4}}. \label{eq:sw-ge}  
   \end{equation}
   Moreover, since we assumed $\north_z \leq \south_z$ and we are in Case~2 
   (so $|L_z|\leq (n-1)/2$) we have
   \begin{equation}
   \NE_z \leq \north_z \leq \floor{|L_z|/2} \leq \floor{(n-1)/4}. \label{eq:nez-le} \end{equation}

Finally, from \eqref{eq:nebound}, the number of incidences in each cell of the north-eastern column is at most
   \begin{align*}
    \frac{\NE+\NE_z}{2} & \leq \frac{n-\SW+\NE_z}{2} & \text{by the Claim above} \\
        & \leq \frac{n-\floor*{\strut(\floor{5n/3}-2)/4}+\NE_z}{2} & \text{by \eqref{eq:sw-ge}} \\
        & \leq \frac{n-\floor*{\strut(\floor{5n/3}-2)/4}+ \floor*{(n-1)/4}}{2} & \text{by \eqref{eq:nez-le}} \\
        & \leq \ceil{5n/12}.
   \end{align*}
   For the last inequality, define $f(n) \coloneqq \frac{n-\floor*{\strut(\floor{5n/3}-2)/4}+ \floor*{(n-1)/4}}{2} - \ceil{5n/12}$.  We need to show that $f(n) \leq 0$ for all $n\in \Nats$.  Observe that $f(n)=f(n+12)$ for any $n$. Hence, it suffices to show that $f(n) \leq 0$ for all integer $n$ from $0$ to $11$, which can be verified by a straightforward computation.  
     \qedhere
\end{proof}

\section{Conclusions}
\label{sec:concl}

Two obvious open problems remain:
\begin{itemize}
    \item Close the gaps between the upper and the lower bounds for axis-parallel lines.
    \item Answer the same question for lines with arbitrary orientations in $\Reals^3$: given a set $L$ of $n$~lines, minimize the number of lines meeting each open cell of the arrangement $\A(H)$ formed by a set~$H$ of three planes.  
    A simple calculation shows that if $L$  is in general position, then
    at least one of the open cells of $\A(H)$ meets at least $4(n-3)/8\approx n/2$
    lines from $L$, since any plane can contain or be parallel to at most one of the lines of~$L$.
    Can we prove larger lower bounds, and what upper bounds
    can we obtain?
\end{itemize}

\bibliographystyle{plainurl}
\bibliography{refs}
\end{document}